\documentclass[11pt]{article}

\newcommand{\version}{long}

\usepackage{algorithm,algpseudocode,amsmath,amssymb,amsthm,complexity,etoolbox,fullpage,graphicx,ifthen,mathtools,url,thmtools,thm-restate}
\usepackage[inline,shortlabels]{enumitem}

\usepackage{complexity}

\newclass{\ioPSPACE}{i.o.\text{-}PSPACE}
\newlang{\Halt}{Halt}
\usepackage{etoolbox,ifthen,tabulary}

\newcommand{\ifft}{\ifundef{\shortiff}{if and only if }{iff }}
\newcommand{\spc}{{ }}

\newcommand{\inpuborpriv}[2]{\ifdef{\priv}{#2}{#1}}
\newcommand{\inpub}[1]{\inpuborpriv{#1}{}}
\newcommand{\inpriv}[1]{\inpuborpriv{}{#1}}
\newcommand{\indraft}[1]{\ifthenelse{\equal{\version}{draft}}{#1}{}}
\newcommand{\infinal}[1]{\ifthenelse{\equal{\version}{final}}{#1}{Only shown in the final version}}
\newcommand{\inshort}[1]{\ifthenelse{\equal{\version}{short}}{#1}{}}
\newcommand{\inlong}[1]{\ifthenelse{\equal{\version}{long}}{#1}{}}

\newcommand{\classorcat}{\inpuborpriv{class}{category}}

\newcommand{\appref}[1]{Appendix~\ref{app:#1}}
\newcommand{\secref}[1]{Section~\ref{sec:#1}}

\newcommand{\renameenv}[2]{
  \expandafter\let\csname #1#2\expandafter\endcsname
  \csname #1\endcsname
  \expandafter\let\csname end#1#2\expandafter\endcsname
  \csname end#1\endcsname
  \expandafter\let\csname #2\endcsname\relax
  \expandafter\let\csname end#2\endcsname\relax}

\ifundef{\defaultlists}{
  \usepackage[inline,shortlabels]{enumitem}
  \setenumerate[1]{(a),itemsep=0pt,topsep=3pt,parsep=0pt,partopsep=0pt}
  \setenumerate[2]{(i),noitemsep,topsep=3pt,parsep=0pt,partopsep=0pt}
  \setenumerate[3]{(A),noitemsep,topsep=3pt,parsep=0pt,partopsep=0pt}
  \setenumerate[4]{(I),noitemsep,topsep=3pt,parsep=0pt,partopsep=0pt}
  \setitemize{noitemsep,topsep=3pt,parsep=0pt,partopsep=0pt}
  \setdescription{noitemsep,topsep=3pt,parsep=0pt,partopsep=0pt}
  \setlist{noitemsep,topsep=3pt,parsep=0pt,partopsep=0pt}}{}

\newcolumntype{x}[1]{>{\centering\arraybackslash}m{#1}}
\usepackage{algpseudocode,amsmath,amssymb,etoolbox,fixltx2e,mathtools}

\let\eqref\relax

\DeclareFontFamily{U}{mathx}{\hyphenchar\font45}
\DeclareFontShape{U}{mathx}{m}{n}{
      <5> <6> <7> <8> <9> <10>
      <10.95> <12> <14.4> <17.28> <20.74> <24.88>
      mathx10
      }{}
\DeclareSymbolFont{mathx}{U}{mathx}{m}{n}
\DeclareMathSymbol{\bigtimes}{1}{mathx}{"91}

\algnewcommand{\Input}{\item[\textbf{Input:}]}
\algnewcommand{\Output}{\item[\textbf{Output:}]}
\MakeRobust{\Call} 

\newcommand{\comp}{\mathbf{Comp}}
\newcommand{\comptree}{\mathbf{CompTree}}

\newcommand{\bigleafprod}{\bigodot}
\newcommand{\leafprod}{\odot}
\newcommand\restr[2]{{
  \left.\kern-\nulldelimiterspace
  #1
  \vphantom{\big|}
  \right|_{#2}
  }}

\newcommand{\grp}{\mathbf{Grp}}
\newcommand{\graph}{\mathbf{Graph}}

\newcommand{\id}{\mathrm{id}}
\newcommand{\concat}{:}

\newcommand{\abs}[1]{\left\vert#1\right\vert}
\newcommand{\can}{\mathrm{Can}}

\newcommand{\cay}{\mathrm{Cay}}

\renewcommand{\hom}{\mathrm{Hom}}

\newcommand{\iso}{\mathrm{Iso}}

\newcommand{\rank}{\mathrm{rank}}

\newcommand{\setb}[2]{\left\{#1 \;\middle|\; #2\right\}}

\newcommand{\nth}[1]{\ensuremath{{#1}^{\mathrm{th}}}}

\newcommand{\thmref}[1]{Theorem~\ref{thm:#1}}
\newcommand{\lemref}[1]{Lemma~\ref{lem:#1}}

\newcommand{\defref}[1]{Definition~\ref{defn:#1}}

\newcommand{\propref}[1]{Proposition~\ref{prop:#1}}
\newcommand{\corref}[1]{Corollary~\ref{cor:#1}}

\newcommand{\eqref}[1]{(\ref{eq:#1})}

\newcommand{\eps}{\epsilon}

\newcommand{\bmg}{\mathbf{g}}
\newcommand{\bmh}{\mathbf{h}}

\newcommand{\cC}{\mathcal{C}}
\newcommand{\cD}{\mathcal{D}}

\newcommand{\cG}{\mathcal{G}}
\newcommand{\cH}{\mathcal{H}}

\newcommand{\la}{\leftarrow}

\newcommand{\ra}{\rightarrow}

\newcommand{\tril}{\triangleleft}

\newcommand{\trile}{\trianglelefteq}

\usepackage{amsthm}
\makeatletter
\def\thm@space@setup{\thm@preskip=3pt \thm@postskip=3pt}
\makeatother

\makeatletter
\renewenvironment{proof}[1][\proofname]{\par
  \pushQED{\qed}
  \normalfont
  \topsep3pt \partopsep0pt 
  \trivlist
  \item[\hskip\labelsep
        \itshape
    #1\@addpunct{.}]\ignorespaces
  }{
    \popQED\endtrivlist\@endpefalse
    \addvspace{0pt plus 0pt} 
  }
\makeatother

\usepackage{etoolbox,ifthen,thmtools,thm-restate}

\ifundef{\dontnumberwithin}{\declaretheorem[numberwithin=section]{dummy}}{\declaretheorem{dummy}} 
\declaretheorem[sibling=dummy]{theorem}

\declaretheorem[sibling=dummy]{lemma}

\declaretheorem[sibling=dummy]{definition}

\declaretheorem[sibling=dummy]{proposition}
\declaretheorem[sibling=dummy]{corollary}

\ifundef{\defaultthmcontinues}{\renewcommand{\thmcontinues}[1]{}}{}

\inshort{\newcommand{\shortiff}{}}
\inshort{\usepackage[compact]{titlesec}}

\begin{document}

\title{Beating the Generator-Enumeration Bound for $p$-Group Isomorphism\footnote{Preliminary versions of this work appeared as portions of~\cite{wagner2011a} and~\cite{rosenbaum2013a}.}}
\author{David J. Rosenbaum \\ {\small Department of Computer Science \& Engineering} \\ {\small University of Washington} \\ {\small Email: djr@cs.washington.edu} \and Fabian Wagner \\ {\small Universit\"{a}t Ulm} \\ {\small Institut f\"{u}r Theoretische Informatik} \\ {\small Email: fabian.wagner@uni-ulm.de}}
\date{December  6, 2013}

\maketitle
\thispagestyle{empty}

\begin{abstract}
  We consider the \emph{group isomorphism problem}: given two finite groups $G$ and $H$ specified by their multiplication tables, decide if $G \cong H$.  For several decades, the $n^{\log_p n + O(1)}$ generator-enumeration bound (where $p$ is the smallest prime dividing the order of the group) has been the best worst-case result for general groups.  In this work, we show the first improvement over the generator-enumeration bound for $p$-groups, which are believed to be the hard case of the group isomorphism problem.  We start by giving a Turing reduction from group isomorphism to $n^{(1 / 2) \log_p n + O(1)}$ instances of $p$-group composition-series isomorphism.  By showing a Karp reduction from $p$-group composition-series isomorphism to testing isomorphism of graphs of degree at most $p + O(1)$ and applying algorithms for testing isomorphism of graphs of bounded degree, we obtain an $n^{O(p)}$ time algorithm for $p$-group composition-series isomorphism.  Combining these two results yields an algorithm for $p$-group isomorphism that takes at most $n^{(1 / 2) \log_p n + O(p)}$ time.  This algorithm is faster than generator-enumeration when $p$ is small and slower when $p$ is large.  Choosing the faster algorithm based on $p$ and $n$ yields an upper bound of $n^{(1 / 2 + o(1)) \log n}$ for $p$-group isomorphism.
\end{abstract}


\inlong{
  \newpage
  \setcounter{page}{1}}

\section{Introduction}
\label{sec:intro}
In the \emph{group isomorphism problem}, we are given two finite groups $G$ and $H$ with element set $[n]$ as multiplication tables and must decide if $G \cong H$.  It is not known if group isomorphism is in \P\spc while its \NP-completeness would contradict the exponential-time hypothesis~\cite{impagliazzo1999a}; The strongest complexity theoretic upper bound is that solvable-group isomorphism is in $\NP \cap \coNP$ assuming that $\EXP \not\subseteq \ioPSPACE$~\cite{arvind2003a}.  Group isomorphism is Karp reducible to graph isomorphism (\GI)~\cite{hedrlin1966a,hedrlin1969a,miller1979a} but is not known to be \GI-complete and may be easier due to the group structure~\cite{chattopadhyay2010a}.  Therefore if group isomorphism is \NP-complete, then the polynomial hierarchy collapses to the second level~\cite{babai1985a,boppana1987a,goldwasser1989a,goldreich1991a}.

The \emph{generator-enumeration algorithm} for group isomorphism has been known for several decades~\cite{felsch1970a} (cf.~\cite{lipton1977a,miller1978a}).  Given a generating set $S$ for a group $G$ we can test if $G$ is isomorphic to a group $H$ in $n^{\abs{S} + O(1)}$ time by considering all $n^{\abs{S}}$ maps from $S$ into $H$ and checking if any of these extend to an isomorphism.  If $p$ is the smallest prime dividing the order of $G$, then $G$ has a generating set of size at most $\log_p n$ which can be found in $n^{\log_p n + O(1)}$ by brute force.  This yields an $n^{\log_p n + O(1)}$ upper bound for group isomorphism.

Lipton, Snyder and Zalcstein~\cite{lipton1977a} proved the related result that group isomorphism can be decided using $O(\log^2 n)$ space.

Faster algorithms have been obtained for various special cases.  Lipton, Snyder and Zalcstein~\cite{lipton1977a}; Savage~\cite{savage1980a}; Vikas~\cite{vikas1996a}; and Kavitha~\cite{kavitha2007a} all showed polynomial-time algorithms for Abelian groups.  Le Gall~\cite{legall2008a} gave a polynomial-time algorithm for groups consisting of a semidirect product of an Abelian group with a cyclic group of coprime order; this was extended to a class of groups with a normal Hall subgroup by Qiao, Sarma and Tang~\cite{qiao2011a}.  Babai and Qiao~\cite{babai2012a} showed that testing isomorphism of groups with Abelian Sylow towers is in polynomial time.  Babai, Codenotti, Grochow and Qiao~\cite{babai2011a} showed an $n^{O(\log \log n)}$ time algorithm for the class of groups with no normal Abelian subgroups; the runtime was later improved to polynomial by Babai, Codenotti and Qiao~\cite{codenotti2011a,babai2012b}.


Our main result is an algorithm that is faster than the generator-enumeration algorithm for $p$-groups, which are believed to be the hard case of group isomorphism~\cite{babai2011a,codenotti2011a,babai2012a}.  Before our work, it was a longstanding open problem to obtain an $n^{(1 - \eps) \log_p n + O(1)}$ algorithm where $\eps > 0$~\cite{lipton2011a}.

\begin{restatable}{theorem}{pgroupiso}
  \label{thm:p-group-iso}
  $p$-group isomorphism is decidable in $n^{\min\{(1 / 2) \log_p n + O(p),\;\log_p n\}}$ time.  In particular, \\ $n^{(1 / 2) \log_p n + O(\log n / \log \log n)}$ and $n^{(1 / 2) \log n + O(1)}$ are upper bounds on the complexity.
\end{restatable}

The first step in our algorithm reduces group isomorphism to many instances of composition-series isomorphism.  (Two composition series are isomorphic if there exists an isomorphism that maps each subgroup in the first series to the corresponding subgroup in the second series.)

\begin{restatable}{theorem}{groupredcomp}
  \label{thm:group-red-comp}
  Testing isomorphism of two groups $G$ and $H$ is $n^{(1 / 2) \log_p n + O(1)}$ time Turing reducible to testing isomorphism of composition series for $G$ and $H$ where $p$ is the smallest prime dividing the order of the group.
\end{restatable}

This bound can be proved by counting the number of composition series using a simple argument.  We are grateful to Laci Babai for pointing this out as it simplifies our algorithm.

Our second step is reduce $p$-group composition-series isomorphism to testing isomorphism of graphs of degree $p + O(1)$.  We accomplish this by constructing a tree with a node for each coset for each of the intermediate subgroups in the composition series.  By adding certain gadgets that encode the multiplication table of the group, we show that composition series for two $p$-groups are isomorphic \ifft the graphs resulting from this construction are isomorphic.  By applying an $n^{O(d)}$ time algorithm~\cite{luks1982a,babai1983a,babai1983b} for testing isomorphism of graphs of degree at most $d$, we obtain an $n^{O(p)}$ time algorithm for $p$-group composition-series isomorphism.  Combining this result with our reduction from group isomorphism to composition-series isomorphism yields an $n^{(1 / 2) \log_p n + O(p)}$ time algorithm for testing isomorphism of $p$-groups.  Combining this algorithm with the generator-enumeration algorithm completes the proof of \thmref{p-group-iso}.


The canonical form of a class of objects is a function that maps each object to a unique representative of its isomorphism class.  Since canonical forms of graphs of degree at most $d$ can be computed in $n^{O(d)}$ time~\cite{luks1982a,babai1983a}\footnote{Luks showed that there is a faster $n^{O(d / \log d)}$ algorithm for testing isomorphism of graphs of degree at most $d$~\cite{babai1983b} but it does not improve any of our results.}, \thmref{p-group-iso} can be modified to perform \emph{$p$-group canonization} in the same complexity bound.  If $p \leq \alpha$ is small, we compute the canonical form of the graph that arises from each choice of composition series and choose the one that comes first lexicographically.  A canonical multiplication table for the $p$-group is then recovered from this canonical form.  When $p > \alpha$, we use a variant of the generator-enumeration algorithm that performs group canonization.

The framework of reducing group isomorphism to composition-series isomorphism and then to low-degree graph isomorphism was proposed by the second author~\cite{wagner2011a}; the results in this paper appeared in preliminary form in~\cite{rosenbaum2013a}.  Recent papers by the first author extend all of our results to solvable groups~\cite{rosenbaum2013d} using Hall's theory of Sylow bases~\cite{hall1938a} and introduce collision detection arguments that yield an $n^{(1 / 2) \log_p n + O(1)}$ algorithm for testing isomorphism of general groups~\cite{rosenbaum2013b}.  Combined with the canonization machinery just mentioned, these collision detection arguments to reduce the $1 / 2$ in exponent of \thmref{p-group-iso} to $1 / 4$~\cite{rosenbaum2013b}.  While the generator-enumeration algorithm and \thmref{p-group-iso} both require only polynomial space, both of these speedups come at the cost of requiring space comparable to the runtime.  

We start with group \inpuborpriv{theory}{and category theory} background in \secref{background}.  In \secref{comp-red}, we reduce group isomorphism to composition-series isomorphism.  In \secref{graph-red}, we present the reduction from $p$-group isomorphism to low-degree graph isomorphism.  In \secref{p-algorithms}, we derive our algorithms for $p$-group isomorphism.

\section{\inpuborpriv{Group-theory b}{B}ackground}
\label{sec:background}
\inpriv{\subsection{Group theory}}
In this \inpriv{sub}section, we review some basic group theory that will be used in the paper; we omit the proofs which may be found in group theory and algebra texts~\cite{hungerford1974a,rotman1995a,robinson1996a,lang2002a,holt2005a,rose2009a,artin2010a,wilson2010a,roman2011a}.  Readers familiar with group theory may to skip this section.

Let $G$ and $H$ be groups.  Throughout this paper, we assume that all groups are finite and let $n = \abs{G}$.  \inpriv{The \emph{rank} of a group $G$ (denoted $\rank(G)$) is the size of the smallest generating set for $G$.  }We let $\iso(G, H)$ denote the set of all isomorphisms from $G$ to $H$.  For a prime $p$, a \emph{$p$-group} is a group of order a power of $p$.  The \emph{conjugation} of $x$ by $g$ is written as $x^g = g x g^{-1}$; the bijections $\iota_g : G \ra G : x \mapsto x^g$ are the \emph{inner automorphisms} of $G$.  A \emph{normal} subgroup $N$ of $G$ (denoted $N \trile G$) is a subgroup that is closed under conjugation by elements of $G$.  The \emph{left coset} of an element $g \in G$ for a subgroup $H \leq G$ is the set $g H = \setb{g h}{h \in H}$.  Since we do not use right cosets in this paper, we will simply refer to left cosets as \emph{cosets}.  If $H \leq G$, then $G / H$ denotes the group of cosets of $H$ in $G$.

A \emph{subnormal series} $S$ of a group $G$ is a tower of subgroups $G_0 = 1 \tril \cdots \tril G_m = G$.  The groups $G_{i + 1} / G_i$ are called the \emph{factors} of $S$.  A \emph{composition series} of a group $G$ is a maximal subnormal series for $G$.  For a $p$-group, every composition factor is isomorphic to the cyclic group of order $p$.

\inpriv{\subsection{Category theory}
It will be convenient for us to use a small amount of terminology from category theory.  A \emph{category} is a class of objects together with a collection of \emph{morphisms} between pairs of objects such that there is an identity morphism and composition of morphisms is associative.  The class of morphisms between two objects $C$ and $C'$ in a category $\cC$ is denoted $\hom(C, C')$.  An example of a category is the class of all groups together with the homomorphisms between the groups.  An invertible morphism is called an \emph{isomorphism}.  A \emph{functor} $F : \cC \ra \cD$ is a mapping from the category $\cC$ to the category $\cD$ that sends each object $C$ in $\cC$ to an object $F(C)$ in $\cD$.  For each morphism $f : C \ra C'$, there is a morphism $F(f) : F(C) \ra F(C')$; moreover, $F$ is required to respect the identity morphism and composition.  The functor $F$ is \emph{full} if for each pair of objects $C$ and $C'$ in $\cC$, the map $F_{C, C'} : \hom(C, C') \ra \hom(F(C), F(C'))$ is surjective.  The functor $F$ is \emph{faithful} if each $F_{C, C'}$ is injective.  A \emph{natural isomorphism} $\nu : F \ra G$ between two functors $F : \cC \ra \cD$ and $G : \cC \ra \cD$ that maps each $C$ in $\cC$ to an isomorphism $\nu_C : F(C) \ra G(D)$ such that $\nu_{C'} F(f) = G(f) \nu_C$ for each morphism $f : C \ra C'$.  Two functors $F : \cC \ra \cD$ and $G : \cD \ra \cC$ are a \emph{category equivalence} if $FG$ and $GF$ are naturally isomorphic to the identity functors $I_{\cD}$ and $I_\cC$.}

\section{Reducing group isomorphism to composition-series isomorphism}
\label{sec:comp-red}
In this section, we prove an upper bound on the number of composition series for a group and provide a simple method for enumerating all such composition series.  In earlier versions of this work, a more complex construction was used to enumerate all composition series within a particular class and an upper bound was proved on the size of this class of composition series.  However, Laci Babai pointed out that the upper bound actually holds for the class of all composition series.  This allows a much simpler argument to be employed.

\begin{lemma}
  \label{lem:comp-bound}
  Let $G$ be a group.  Then the number of composition series for $G$ is at most $n^{(1 / 2) \log_p n + O(1)}$ where $p$ is the smallest prime dividing the order of the group.  Moreover, one can enumerate all composition series for $G$ in $n^{(1 / 2) \log_p n + O(1)}$ time.
\end{lemma}

\begin{proof}
  We show that one can enumerate a class of chains that contains all maximal chains of subgroups in $n^{\log_p n + O(1)}$ time.  Since every maximal chain of subgroups contains at most one composition series as a subchain, this suffices to prove the result.

  We start by choosing the first nontrivial subgroup in the series.  Each of these is generated by a single element so there are at most $n$ choices.  If we have a chain $G_0 = 1 < \cdots < G_k$ of subgroups of $G$, then the next subgroup in the chain can be chosen in at most $\abs{G / G_k}$ ways since different representatives of the same coset generate the same subgroup.  Since each $\abs{G_{i + 1}} \geq p \abs{G_i}$, we see that the number of choices $\abs{G / G_k}$ for $G_{k + 1}$ is at most $n / p^k$.  Therefore, the total number of choices required to construct a chain of subgroups in this manner is at most

  \begin{align*}
    \prod_{k = 0}^{\lfloor \log_p n \rfloor - 1} (n / p^k) & \leq p^{\sum_{k = 0}^{\lceil \log_p n \rceil} k} \\
                                                        {} & = p^{(1 / 2) \log_p^2 n + O(\log_p n)} \\
                                                        {} & \leq n^{(1 / 2) \log_p n + O(1)}
  \end{align*}

Since the set of subgroup chains enumerated by this process includes all maximal chains of subgroups, the result follows. 
\end{proof}

We say that two composition series $G_0 = 1 \tril \cdots \tril G_m = G$ and $H_0 = 1 \tril \cdots \tril H_{m'} = H$ are isomorphic if there exists an isomorphism $\phi : G \ra H$ such that each $\phi[G_i] = H_i$.  It is now very easy to obtain the Turing reduction from group isomorphism to composition series isomorphism.

\groupredcomp*

\begin{proof}
  Let $G$ and $H$ be groups.  Fix a composition series $S$ for $G$.  If $G \cong H$, then some composition series $S'$ for $H$ will be isomorphic to $S$.  Thus, testing isomorphism of $G$ and $H$ reduces to testing if $S$ is isomorphic to some composition series for $S'$.  The result is then immediate from \lemref{comp-bound}.
\end{proof}

The reduction also applies to reducing group canonization to composition series canonization.





\begin{theorem}
  \label{thm:group-red-comp-can}
  Computing the canonical form of a group is $n^{(1 / 2) \log_p n + O(1)}$ time Turing reducible to computing canonical forms of composition series for the group where $p$ is the smallest prime dividing the order of the group.
\end{theorem}

\begin{proof}
  Let $G$ be a group.  We use \lemref{comp-bound} to enumerate all of the at most $n^{(1 / 2) \log_p n + O(1)}$ composition series $S$ for $G$ and compute the canonical form of each one.  From each such canonical form, we extract the multiplication table and define $\can_{\grp}(G)$ to be the lexicographically least matrix among all such multiplication tables.  Since two groups $G$ and $H$ have isomorphic composition series $S$ and $S'$ \ifft $G \cong H$, it follows that $\can_{\grp}$ is a canonical form.
\end{proof}

\section{Composition-series isomorphism and canonization}
\label{sec:graph-red}
In this section, we reduce composition-series isomorphism to low-degree graph isomorphism.  We also extend the reduction to perform composition-series canonization instead of isomorphism testing.

First, we review some basic concepts for graphs.  A \emph{colored graph} is a graph that associates each vertex with a given color.  Two colored graphs are isomorphic if there is a bijection between their vertex sets that respects the edges and maps each node to a node of the same color.  We shall make use of the following result.

\begin{theorem}[Babai and Luks~\cite{babai1983a,babai1983b}]
  \label{thm:const-deg-can}
  Canonization of colored graphs of degree at most $d$ is in $n^{O(d)}$ time.
\end{theorem}

\subsection{Isomorphism testing}
To test if two composition series are isomorphic, we construct a tree by starting with the whole group $G$ and decomposing it into its cosets $G / G_{m - 1}$; we then further decompose each coset in $G / G_{m - 1}$ into the cosets $G / G_{m - 2}$ that it contains.  This process is repeated until we reach the trivial group $G_0 = 1$.  We make this precise with the following definition.

\begin{definition}
  \label{defn:TS}
  Let $G$ be a group and consider the composition series $S$ given by the subgroups $G_0 = 1 \tril \cdots \tril G_m = G$.  Then $T(S)$ is defined to be the root tree whose nodes are $\bigcup_i G / G_i$.  The root node is $G$.  The leaf nodes are $\{x\} \in G / 1$ which we identify with $x \in G$.  For each node $x G_{i + 1} \in G / G_{i + 1}$, there is an edge to each $y G_i$ such that $y G_i \subseteq x G_{i + 1}$.
\end{definition}

We now use this tree to define a graph that encodes the multiplication table of $G$.  The idea is to attach a multiplication gadget to the nodes $x, y, z \in G$ for each entry $x y = z$ in the multiplication table.  If we did this naively, each node $x \in G$ would have degree $\Omega(n)$.  We address this problem by defining a variant of the rooted product~\cite{godsil1978a} which we call a leaf product.  Let $T_1$ and $T_2$ be rooted trees.  The \emph{leaf product} of $T_1$ and $T_2$ (denoted $T_1 \leafprod T_2$) is the tree obtained by creating a copy of $T_2$ for each leaf node of $T_1$ and identifying the root of each copy with one of the leaf nodes.  We denote by $L(T)$ the set of leaves of the tree $T$.

\begin{definition}
  \label{defn:leaf-prod}
  Let $T_1$ and $T_2$ be trees rooted at $r_1$ and $r_2$.  Then the leaf product $T_1 \leafprod T_2$ is the tree rooted at $r_1$ with vertex set

  \begin{equation*}
    V(T_1) \cup \setb{(x, y)}{x \in L(T_1) \text{ and } y \in V(T_2) \setminus \{r_2\}}
  \end{equation*}

  The set of edges is

  \begin{align*}
    E(T_1) & \cup \setb{(x, (x, y))}{x \in L(T_1) \text{ and } (r_2, y) \in E(T_2)} \\
        {} & \cup \setb{((x, y), (x, z))}{x \in L(T_1) \text{ and } (y, z) \in E(T_2) \text{ where } y, z \not= r_2}
  \end{align*}
\end{definition}





Leaf products are non-commutative but are associative if we identify the tuples $(x, (y, z))$, $((x, y), z)$ with $(x, y, z)$ in the vertex set.  (This is same sense in which cross products are associative.)  We shall make this identification from now on as it simplifies our notation.

Since we will need to consider isomorphisms of leaf products of trees, it is also useful to define leaf products of tree isomorphisms.

\begin{definition}
  \label{defn:leaf-prod-iso}
  For each $1 \leq i \leq k$, let $T_i$ and $T_i'$ be trees rooted at $r_i$ and $r_i'$ and let $\phi_i : T_i \ra T_i'$ be an isomorphism.  Then the leaf product $\bigleafprod_{i = 1}^k \phi_i : \bigleafprod_{i = 1}^k T_i \ra \bigleafprod_{i = 1}^k T_i'$ sends each $(x_1, \ldots, x_j)$ to $(\phi_1(x_1), \ldots, \phi_j(x_j))$ where each $x_i \in L(T_i)$ for $i < j$, $x_j \in V(T_j) \setminus \{r_j\}$ and $j \leq k$.
\end{definition}

For a bijection $\phi$ between the leaves of two trees, we shall use the notation $\hat \phi$ to denote the unique isomorphism between the trees to which $\phi$ extends (when such an isomorphism exists).  The following extension of leaf products is convenient.  For each $1 \leq i \leq k$, let $\phi_i$ be a bijection from the leaves of $T_i$ to the leaves of $T_i'$ that extends uniquely to an isomorphism $\hat \phi : T_i \ra T_i'$.  Then we define $\bigleafprod_{i = 1}^k \phi_i = \bigleafprod_{i = 1}^k \hat \phi_i$.

It is easy to see that $\bigleafprod_{i = 1}^k \phi_i$ is an isomorphism from $\bigleafprod_{i = 1}^k T_i$ to $\bigleafprod_{i = 1}^k T_i'$.

\begin{proposition}
  \label{prop:leaf-prod-well}
  For each $1 \leq i \leq k$, let $T_i$ and $T_i'$ be rooted trees and let $\phi_i$ be a bijection between the leaves of $T_i$ and $T_i'$ such that $\phi_i$ extends uniquely to an isomorphism from $T_i$ to $T_i'$.  Then $\bigleafprod_{i = 1}^k \phi_i : \bigleafprod_{i = 1}^k T_i \ra \bigleafprod_{i = 1}^k T_i'$ is a well-defined isomorphism.
\end{proposition}

As we mentioned earlier, simply attaching multiplication gadgets to the leaves of the tree $T(S)$ would result in a tree of large degree.  We resolve this problem by considering the tree $T(S) \leafprod T(S)$ instead.  We show how to construct multiplication gadgets so that each of the $n^2$ leaf nodes is involved in only a constant number of edges.  This causes the resulting graph to have degree $p + O(1)$ when $G$ is a $p$-group.  The details of this construction are described in the following definition.

\begin{definition}
  \label{defn:X-S}
  Let $G$ be a group and let $S$ be the composition series $G_0 = 1 \tril \cdots \tril G_m = G$.  Let $M$ be the tree with a root connected to three nodes $\la$, $\ra$ and $=$ with colors ``left'', ``right'' and ``equals'' respectively.  To construct $X(S)$, we start with the tree $T(S) \leafprod T(S) \leafprod M$ and connect multiplication gadgets to the leaf nodes.  For each $x, y \in G$, we create the path $((x, y, \la), (y, x, \ra), (x y, y, =))$.  The nodes other than the leaf nodes in $X(S)$ are colored ``internal.''
\end{definition}

The graph $X(S)$ is a \emph{cone graph}; that is, a rooted tree with edges between nodes at the same level.  We call the edges that form the tree in a cone graph \emph{tree edges} and the edges between nodes at the same level \emph{cross edges}.

Our next goal is to show that two composition series $S$ and $S'$ are isomorphic \ifft $X(S)$ and $X(S')$ are isomorphic.  Let $\comp$ be the \classorcat of composition series for finite groups \inpuborpriv{and}{and isomorphisms between them;} let $\comptree$ be the \classorcat of graphs that are isomorphic to a graph $X(S)$ for some composition series $S$\inpriv{ and isomorphisms between such graphs}.  For each pair of composition series $S$ and $S'$ and each isomorphism $\phi : S \ra S'$, we overload the symbol $X$ from \defref{X-S} by defining $X(\phi) : X(S) \ra X(S')$ to be $\phi \leafprod \phi \leafprod \id_M$\inpriv{ (thus obtaining a functor)}.

We seek to show that \inpuborpriv{for two composition series $S$ and $S'$, the map $X_{S, S'} : \iso(S, S') \ra \iso(X(S), X(S'))$ given by $\phi \mapsto X(\phi)$ is surjective and}{$X$ is a full functor that} can be evaluated in polynomial time.  We note that in particular, this result shows that $X$ can be used to reduce composition series isomorphism to testing isomorphism of the resulting graphs.  We start by showing that any isomorphism between $S$ and $S'$ maps to an isomorphism between $X(S)$ and $X(S')$.

\begin{lemma}
  \label{lem:X-functor}
  \inpuborpriv{For each pair of composition series $S$ and $S'$, $X_{S, S'} : \iso(S, S') \ra \iso(X(S), X(S'))$ is well-defined.}{$X : \comp \ra \comptree$ is a functor.}
\end{lemma}

\begin{proof}
  Let\inpriv{ $S$ and $S'$ be composition series} $G_0 = 1 \tril \cdots \tril G_m = G$ and $H_0 = 1 \tril \cdots \tril H_m = H$\inpub{ be the subgroup chains for the composition series $S$ and $S'$} and let $\phi : S \ra S'$ be an isomorphism.  We can view $\phi$ as a bijection from the leaves of $T(S)$ to the leaves of $T(S')$.  Since each $\phi[G_i] = H_i$, we see that $\phi$ extends to a unique isomorphism $\hat \phi : T(S) \ra T(S')$.  By \propref{leaf-prod-well}, $X(\phi) : T(S) \leafprod T(S) \leafprod M \ra T(S') \leafprod T(S') \leafprod M$ is a tree isomorphism.  Then by \defref{X-S}, we just need to show that $X(\phi)$ respects the cross edges representing the multiplication gadgets.

  Let $x, y \in G$.  Then $X(S)$ contains the path $((x, y, \la), (y, x, \ra), (x y, y, =))$.  In $H$, $\phi(x) \phi(y) = \phi(x y)$ so $X(S')$ contains the path $((\phi(x), \phi(y), \la), (\phi(y), \phi(x), \ra), (\phi(x y), \phi(y), =))$.  By definition, we see that $X(\phi)$ maps the path $((x, y, \la), (y, x, \ra), (x y, y, =))$ in $X(S)$ to the path $((\phi(x), \phi(y), \la), (\phi(y), \phi(x), \ra), (\phi(x y), \phi(y), =))$ in $X(S')$.  Since $X(S)$ and $X(S')$ have equal numbers of cross edges, it follows that $X(\phi) : X(S) \ra X(S')$ is an isomorphism.\inpriv{

  Finally, if $S''$ is a composition series and $\psi$ is an isomorphism from $S'$ to $S''$ then $X(\psi \phi) = X(\psi) X(\phi)$ and $X(\id_S) = \id_{X(S)}$.  Thus, $X$ is a functor.}
\end{proof}

Next, we show that \inpuborpriv{each $X_{S, S'}$ is surjective}{$X$ is a full functor}.  This is more difficult and is accomplished by the next two results.  We first show that every isomorphism from $X(S)$ to $X(S')$ can be expressed as a leaf product.


\begin{lemma}
  \label{lem:iso-decomp}
  Let $S$ and $S'$ be composition series for the groups $G$ and $H$ and let $\theta : X(S) \ra X(S')$ be an isomorphism.  Define $\phi : G \ra H$ to be $\restr{\theta}{G}$.  Then

  \begin{enumerate}
  \item $\theta = \phi \leafprod \phi \leafprod \id_M$ and
  \item $\phi : S \ra S'$ is an isomorphism.
  \end{enumerate}
\end{lemma}

\begin{proof}
  First, we prove part (a).  It is clear that $\phi$ is a bijection between $G$ and $H$ that extends uniquely to an isomorphism from $T(S)$ to $T(S')$.  Let $x, y \in G$.  We will say a path from $x$ to $y$ is \emph{left-right} if it  starts at $x$, moves to a node colored ``left'' along tree edges (away from the root), follows a cross edge to a node colored ``right'' and then moves to $y$ along tree edges (towards the root).  Since the only cross edge in $X(S)$ colored $(\text{``left''}, \text{``right''})$ between the subtrees of $T(S) \leafprod T(S) \leafprod M$ rooted at $x$ and $y$ is $((x, y, \la), (y, x, \ra))$, there is exactly one left-right path from $x$ to $y$.  We denote this path by $P(x, y)$.

  Since $\theta$ maps the root of $X(S)$ to the root of $X(S')$, $\theta$ maps left-right paths to left-right paths.  Therefore, $\theta$ sends $P(x, y)$ to $P(\phi(x), \phi(y))$ so the node $(x, y, \la)$ in $X(S)$ is mapped to the node $(\phi(x), \phi(y), \la)$ in $X(S')$.

  For part (b), we let $x, y, z \in G$ such that $xy = z$.  This multiplication rule is represented in $X(S)$ by the path $((x, y, \la), (y, x, \ra), (z, y, =))$.  By part (a), we know that $\theta$ maps this path to $((\phi(x), \phi(y), \la), (\phi(y), \phi(x), \ra), (\phi(z), \phi(y), =))$.  This implies that $\phi(x) \phi(y) = \phi(z)$ in $H$ so that $\phi$ is an isomorphism from $G$ to $H$.

  Let $G_0 = 1 \tril \cdots \tril G_m = G$ and $H_0 = 1 \tril \cdots \tril H_m = H$ be the chains of subgroups in the composition series $S$ and $S'$.  It remains to show that each $\phi[G_i] = H_i$.  Since $\phi$ is an isomorphism from $G$ to $H$, it follows that $\phi(1) = 1$.  This implies that $\theta$ maps each node $G_i$ in $X(S)$ to the node $H_i$ in $X(S')$.  Then because the elements of $G_i$ correspond precisely to those nodes $x \in G$ such that $x$ is a descendant of the node $G_i$ in $T(S) \leafprod T(S) \leafprod M$, it follows that $\phi[G_i] = H_i$.  Thus, $\phi$ is an isomorphism from $S$ to $S'$.
\end{proof}

\begin{theorem}
  \label{thm:X-fff}
  \inpuborpriv{For each pair of composition series $S$ and $S'$, $X_{S, S'}$ is a bijection}{$X : \comp \ra \comptree$ is a fully faithful functor and can be evaluated in polynomial time}.\inpub{  Moreover, both $X(S)$ and $X(\phi)$ where $\phi \in \iso(S, S')$ can be computed in polynomial time.}
\end{theorem}

\begin{proof}
  Combining Lemmas~\ref{lem:X-functor} and~\ref{lem:iso-decomp} shows that \inpuborpriv{each $X_{S, S'}$ is surjective}{$X$ is a full functor}.  To see that it is \inpuborpriv{injective}{faithful}, we note that if $\phi, \psi \in \iso(S, S')$ and $X(\phi) = X(\psi)$ then $\phi \leafprod \phi \leafprod \id_M = \psi \leafprod \psi \leafprod \id_M$ so $\phi = \psi$.  Since $X$ is defined in terms of leaf products and leaf products can be evaluated in polynomial time, $X$ can also be evaluated in polynomial time.
\end{proof}

The correctness of our reduction follows.

\begin{corollary}
  \label{cor:p-red-cor}
  Let $S$ and $S'$ be composition series.  Then $S \cong S'$ \ifft $X(S) \cong X(S')$.
\end{corollary}

In order to obtain an efficient algorithm for $p$-group composition-series isomorphism, we must show that the degree of the graph is not too large.

\begin{lemma}
  \label{lem:alpha-graph}
  Let $G$ be a group with a composition series $S$ such that $\alpha$ is an upper bound for the order of any factor.  Then the graph $X(S)$ has degree at most $\max\{\alpha + 1, 4\}$ and size $O(n^2)$.
\end{lemma}

\begin{proof}
  The tree $T(S)$ has size $O(n)$ and degree $\alpha + 1$ while the tree $M$ has size $4$ and degree $3$.  Therefore $T(S) \leafprod T(S) \leafprod M$ (and hence $X(S)$) has size $O(n^2)$ and degree $\max\{\alpha + 1, 4\}$.  Adding the edges for the multiplication gadgets in $X(S)$ increases the degrees of the leaves of $T(S) \leafprod T(S) \leafprod M$ to at most $3$, so $X(S)$ also has degree $\max\{\alpha + 1, 4\}$.
\end{proof}

We are now in a position to obtain an algorithm for composition-series isomorphism.

\begin{theorem}
  \label{thm:alpha-comp-iso}
  Let $S$ and $S'$ be composition series such that $\alpha$ is an upper bound for the order of any factor.  Then we can test if $S \cong S'$ in $n^{O(\alpha)}$ time.
\end{theorem}

\begin{proof}
  We can compute the graphs $X(S)$ and $X(S')$ in polynomial time.  By \corref{p-red-cor}, $S \cong S'$ \ifft $X(S) \cong X(S')$.  By \lemref{alpha-graph}, the number of nodes in $X(S)$ is $O(n^2)$ and the degree is at most $\max\{\alpha + 1, 4\} = O(\alpha)$.  Then we can test if $X(S) \cong X(S')$ in $n^{O(\alpha)}$ time using the bounded-degree graph isomorphism algorithm from \thmref{const-deg-can}.
\end{proof}

\subsection{Canonization}
We also show how to compute canonical forms of composition series.  This result is also useful for further improving the efficiency of the algorithm for $p$-group isomorphism (see~\cite{rosenbaum2013b}).  Our high-level strategy for constructing a canonical form for a composition series $S$ is to compute the canonical form of the graph $X(S)$.  We then reconstruct a composition series $Y(\can_{\graph}(X(S)))$ isomorphic to $S$ by inspecting the structure of $\can_{\graph}(X(S))$.

\begin{definition}
  \label{defn:Y-A}
  For each composition series $S$ for a group $G$ and a graph $A \cong X(S)$,  we fix an arbitrary isomorphism $\pi : X(S) \ra A$.  We define $Y(A)$ to be the composition series $\pi[1] \tril \cdots \tril \pi[G]$ for the group with elements $\pi[G]$, where we define $\pi(x) \pi(y) = \pi(z)$ if there exists a path $(a_{\pi(x)}, a_{\pi(y)}, a_{\pi(z)})$ colored $(\text{``left''}, \text{``right''}, \text{``equals''})$, such that $a_{\pi(x)}$, $a_{\pi(y)}$ and $a_{\pi(z)}$ are descendants of $x$, $y$ and $z$ in the image of the tree $T(S) \leafprod T(S) \leafprod M$ under $\pi$.

  For each pair of composition series $S$ and $S'$ for groups $G$ and $H$, graphs $A \cong X(S)$ and $A' \cong X(S')$, let $\pi : X(S) \ra A$ and $\pi' : X(S') \ra A'$ be the fixed isomorphisms chosen above.  Then we define $Y(\theta) : \pi[G] \ra \pi'[H]$ to be $\restr{\theta}{\pi[G]}$.
\end{definition}

First, we need to show that each $Y(A)$ is well-defined.

\begin{lemma}
  \label{lem:Y-well-def}
  Let $S$ be a composition series, let $A$ be a graph and let $\pi : X(S) \ra A$ be an isomorphism.  Then $Y(A)$ is well-defined composition series that can be computed in polynomial time and $Y(\pi)$ is an isomorphism from $S$ to $Y(A)$.
\end{lemma}

\begin{proof}
  Let $G_0 = 1 \tril \cdots \tril G_m = G$ be the subgroup chain for $S$.  We note that the height of $T(S) \leafprod T(S) \leafprod M$ is $2 m + 1$ where $m$ is the composition length of $S$.  Now, $G$ is the group consisting of the elements at a distance of $m$ from the root so $\pi[G]$ is independent of which isomorphism $\pi : X(S) \ra A$ we consider.  Moreover, we can compute $\pi[G]$ in polynomial time.  For each $x, y, z \in G$, $xy = z$ \ifft there exists a path $((x, y, \la), (y, x, \ra), (z, y, =))$ in $X(S)$.  Equivalently, $xy = z$ \ifft there exists a path $(a_x, a_y, a_z)$ colored $(\text{``left''}, \text{``right''}, \text{``equals''})$ where $a_x$, $a_y$ and $a_z$ are descendants of $x$, $y$ and $z$ in $T(S) \leafprod T(S) \leafprod M$.

  Consider the set of elements $\pi[G]$.  For each $\pi(x), \pi(y), \pi(z) \in \pi[G]$, define $\pi(x) \pi(y) = \pi(z)$ \ifft there exists a path $(a_{\pi(x)}, a_{\pi(y)}, a_{\pi(z)})$ colored $(\text{``left''}, \text{``right''}, \text{``equals''})$ where $a_{\pi(x)}$, $a_{\pi(y)}$ and $a_{\pi(z)}$ are descendants of $\pi(x)$, $\pi(y)$ and $\pi(z)$ in the image of $T(S) \leafprod T(S) \leafprod M$ under $\pi$.  Then $\pi[G]$ is a group that we can compute in polynomial time and $Y(\pi)$ is a group isomorphism from $G$ to $\pi[G]$.

  Now, for each $G_i$, $\pi[G_i]$ consists of the nodes in $\pi[G]$ that are descendants of the node $\pi(G_i)$.  Each node $\pi(G_i)$ is the node on the path from the root of $A$ to $\pi(1)$ at distance $m - i$ from the root.  The node $\pi(1)$ is the identity of the group $\pi[G]$ and can therefore be found by inspecting the multiplication rules of $\pi[G]$.  Thus, we can compute each set of nodes $\pi[G_i]$ in polynomial time independently of $\pi$.  This yields a composition series $\pi[1] \tril \cdots \tril \pi[G]$ that does not depend on the choice of $\pi$.  From \defref{Y-A}, we see that this composition series is in fact $Y(A)$.  Moreover, $Y(\pi)$ is an isomorphism from $S$ to $Y(A)$.
\end{proof}

As for $X$, we define $Y_{A, A'} : \iso(A, A') \ra \iso(Y(A), Y(A'))$ by $\theta \mapsto Y(\theta)$ for each pair of graphs $A, A' \in \comptree$.  In order to compute canonical forms, we shall need to show that \inpuborpriv{each $Y_{A, A'}$ is surjective}{$Y$ is a full functor} and can be evaluated in polynomial time.

\begin{theorem}
  \label{thm:Y-fff}
  \inpuborpriv{For each pair of graphs $A, A' \in \comptree$, $Y_{A, A'}$ is a bijection}{$Y : \comptree \ra \comp$ is a fully faithful functor and can be evaluated in polynomial time}.\inpub{  Moreover, both $Y(A)$ and $Y(\theta)$ where $\theta \in \iso(Y(A), Y(A'))$ can be computed in polynomial time.}
\end{theorem}

\begin{proof}
  Let $S$ and $S'$ be composition series with chains of subgroups $G_0 = 1 \tril \cdots \tril G_m = G$ and $H_0 = 1 \tril \cdots \tril H_m = H$, let $A \cong X(S)$ and $A' \cong X(S')$ be graphs, and let $\pi : X(S) \ra A$, $\pi' : X(S) \ra A'$ and $\theta : A \ra A'$ be isomorphisms.

  First, we observe that $Y$ respects composition\inpriv{ and the identity}.  Since $\psi = \theta \pi$ is an isomorphism from $X(S)$ to $A'$, \lemref{Y-well-def} implies that $Y(\psi) = Y(\theta) Y(\pi)$ is an isomorphism from $S$ to $Y(A')$ and that $Y(\pi)$ is an isomorphism from $S$ to $Y(A)$.  It follows that $Y(\theta) = Y(\psi) (Y(\pi))^{-1}$ is an isomorphism from $Y(A)$ to $Y(A')$.  Thus, \inpuborpriv{$Y_{A, A'}$ is a well-defined function}{$Y$ is a well-defined functor}.
  
  To show that \inpuborpriv{$Y_{A, A'}$ is bijective}{$Y$ is fully faithful}, we first note that \inpuborpriv{$YX = I_{\comp}$}{$Y_{X(S), X(S')} X_{S, S'} = I_{\comp}$}, which implies that $Y_{X(S), X(S')}$ is surjective.  \lemref{iso-decomp} implies that $Y_{X(S), X(S')}$ is also injective.  Now, for each $\theta : A \ra A'$, we have $\theta = \pi' \rho \pi^{-1}$ for some isomorphism $\rho : X(S) \ra X(S')$.  Therefore, $Y(\theta) = Y(\pi') Y(\rho) Y(\pi^{-1})$.  Since $Y_{X(S), X(S')}$ is a bijection, we see that $Y_{A, A'}$ is also a bijection.\inpriv{  Thus, $Y$ is fully faithful.}  The fact that $Y$ can be evaluated in polynomial time follows from \defref{Y-A} and \lemref{Y-well-def}.
\end{proof}

\thmref{Y-fff} suffices for our purposes.  However, we note that $X$ and $Y$ form a category equivalence \inpuborpriv{when viewed as functors}{(see \appref{X-Y-equiv})}.\inpub{  All of the results for $X$ and $Y$ in this section are special cases of this more general result.}

To devise an algorithm for composition series canonization, we utilize \inpriv{the functors} $X$ and $Y$ together with the canonical form for graphs of bounded degree \thmref{const-deg-can} (which we denote by $\can_{\graph}$).

\begin{theorem}
  \label{thm:alpha-comp-can}
  The map $Y \circ \can_{\graph} \circ X$ is a canonical form for composition series.  If $S$ is a composition series such that $\alpha$ is an upper bound for the order of any factor, then we can compute $\can_{\comp}(S) = (Y \circ \can_{\graph} \circ X)(S)$ in $n^{O(\alpha)}$ time.
\end{theorem}

\begin{proof}
  Let $S$ and $S'$ be composition series.  First, $X(S) \cong \can_{\graph}(X(S))$ by \thmref{X-fff} which implies that $S \cong Y(\can_{\graph}(X(S)))$ by \thmref{Y-fff}.

  If $S \cong S'$ then $X(S) \cong X(S')$ by \corref{p-red-cor} and $\can_{\graph}(X(S)) = \can_{\graph}(X(S'))$ so $Y(\can_{\graph}(X(S))) = Y(\can_{\graph}(X(S')))$.  On the other hand, if $S \not\cong S'$ then $X(S) \not\cong X(S')$ by \thmref{X-fff} and $\can_{\graph}(X(S)) \not\cong \can_{\graph}(X(S'))$ so $Y(\can_{\graph}(X(S))) \not\cong Y(\can_{\graph}(X(S')))$ by \thmref{Y-fff}.  In particular, $Y(\can_{\graph}(X(S))) \not= Y(\can_{\graph}(X(S')))$.  Thus, $Y \circ \can_{\graph} \circ X$ is a canonical form for composition series.

  By Theorems~\ref{thm:X-fff} and~\ref{thm:Y-fff}, $X$ and $Y$ can be evaluated in polynomial time since the graph $\can_{\graph}(X(S))$ has size $O(n^2)$ and degree $\alpha + O(1)$ by \lemref{alpha-graph}.  Then by \thmref{const-deg-can}, computing the canonical form of $X(S)$ takes $n^{O(\alpha)}$ time.
\end{proof}

\section{Algorithms for $p$-group isomorphism and canonization}
\label{sec:p-algorithms}
The intermediate results of Sections~\ref{sec:comp-red} and \ref{sec:graph-red} put us in a position to prove \thmref{p-group-iso}.

\pgroupiso*

\begin{proof}
  Combining Theorems~\ref{thm:group-red-comp} and~\ref{thm:alpha-comp-iso} yields an $n^{(1 / 2) \log_p n + O(p)}$ time algorithm for testing isomorphism of $p$-groups.  On the other hand, every $p$-group has a generating set of size at most $\log_p n$ so the generator-enumeration algorithm runs in $n^{\log_p n + O(1)}$ time for $p$-groups.  Combining these two algorithms shows that $p$-group isomorphism is decidable in $n^{\min\{(1 / 2) \log_p n + O(p),\;\log_p n\}}$ time.

  Let $\alpha = \log n / \log \log n$.  By upper bounding $\min\{(1 / 2) \log_p n + O(p),\;\log_p n\}$ with $(1 / 2) \log_p n + O(p)$ when $p \leq \alpha$ and with $\log_p n$ when $p > \alpha$, we see that $\min\{(1 / 2) \log_p n + O(p),\;\log_p n\}$ is upper bounded by $(1 / 2) \log_p n + O(\log n / \log \log n)$.  The upper bound $(1 / 2) \log n + O(1)$ can be obtained by showing that the maximum of $(1 / 2) \log_p n + O(p)$ for $p \leq \alpha$ is attained at $p = 2$.
\end{proof}

We remark that in the above algorithm relies on the $n^{O(d)}$ algorithm~\cite{babai1983a} for computing canonical forms of graphs of degree $d$ rather than the faster $n^{O(d / \log d)}$ algorithm~\cite{babai1983a,babai1983b} for testing isomorphism of such graphs.  This does not change the result as $\polylog(d)$ factors in the exponent of the graph isomorphism testing procedure require us to chose a different cutoff $\alpha$ in the proof of \thmref{p-group-iso} but do not affect the final result.

We now adapt our algorithm to perform $p$-group canonization.  The main tool we are missing for this result is the ability to compute the canonical form of a $p$-group in $n^{\log_p n + O(1)}$ time.  Given a total order on an alphabet $\Sigma$, define the \emph{standard order} on $\Sigma^*$ by $x \prec y$ if $\abs{x} < \abs{y}$ or $\abs{x} = \abs{y}$ and $x$ comes before $y$ lexicographically.  We adapt the generator-enumeration algorithm to perform canonization using a lemma that orders the elements of a group using a generating set.  We start by defining the ordering.

\begin{definition}
  \label{defn:gen-ord}
  Let $G$ be a group with an ordered generating set $\bmg = (g_1, \ldots, g_k)$.  Define a total order $\prec_{\bmg}$ on $G$ by $x \prec_{\bmg} y$ if $w_{\bmg}(x) \prec w_{\bmg}(y)$ where each $w_{\bmg}(x) = (x_1, \ldots, x_j)$ is the first word in $\{g_1, \ldots, g_k\}^*$ under the standard ordering such that $x = x_1 \cdots x_j$.
\end{definition}

\begin{lemma}
  \label{lem:gen-ord}
  Let $G$ and $H$ be groups with ordered generating sets $\bmg = (g_1, \ldots, g_k)$ and $\bmh = (h_1, \ldots, h_k)$, and let $x, y \in G$.  Then

  \begin{enumerate}
  \item $\prec_{\bmg}$ is a total ordering on $G$.
  \item if $\phi : G \ra H$ is an isomorphism such that each $\phi(g_i) = h_i$, then $x \prec_{\bmg} y$ \ifft $\phi(x) \prec_{\bmh} \phi(y)$.
  \item we can decide if $x \prec_{\bmg} y$ in $O(n \abs{\bmg})$ time.
  \end{enumerate}
\end{lemma}

\begin{proof}
  Let $S = \{g_1, \ldots, g_k\}$.  For part (a), it is clear that $\prec_{\bmg}$ is a total order since $w_{\bmg} : G \ra S^*$ is clearly injective and the standard ordering on $S^*$ is a total order.

  For part (b), consider an isomorphism $\phi : G \ra H$ such that each $\phi(g_i) = h_i$.  Then if $w_{\bmg}(x) = (x_1, \ldots, x_j)$, $w_{\bmh}(\phi(x)) = (\phi(x_1), \ldots, \phi(x_j))$.  Thus, $x \prec_{\bmg} y$ \ifft $w_{\bmg}(x) \prec w_{\bmg}(y)$ \ifft $w_{\bmh}(\phi(x)) \prec w_{\bmh}(\phi(y))$ \ifft $x \prec_{\bmh} y$.

  For part (c), it suffices to show how to compute $w_{\bmg}(x)$ in polynomial time.  Consider the Cayley graph $\cay(G, S)$ for the group $G$ with generating set $S$.  Then the word $w_{\bmg}(x)$ corresponds to the edges in the minimum length path from $1$ to $x$ in $\cay(G, S)$ that comes first lexicographically.  We can find this path in $O(n \abs{\bmg})$ time by visiting the nodes in breadth-first order starting with $1$.  At the \nth{j} stage, we know $w_{\bmg}(y)$ for all $y \in G$ at a distance of at most $j$ from the root.  We then compute $w_{\bmg}(x)$ for each $x$ at a distance of $j + 1$ from the root by selecting the minimal word $w_{\bmg}(x) \concat g_{x, y}$ over all edges $(x, y)$ associated with an element $g_{x, y}$ of $S$.
\end{proof}

We utilize this order to permute the rows and columns of the multiplication table of the group.

\begin{definition}
  \label{defn:gen-mull-table}
  Let $G$ be a group and let $\bmg$ be an ordered generating set for $G$.  We relabel each element of $G$ by its position in the ordering $\prec_{\bmg}$.  We then permute the rows and columns of the resulting multiplication table so that the elements for the rows and columns appear in the order $1, \ldots, n$ and denote the result by $M_{\bmg}$.
\end{definition}

Clearly, $M_{\bmg}$ defines a group isomorphic to $G$.  The following lemma provides a means of adapting the generator-enumeration algorithm to group canonization.

\begin{lemma}
  \label{lem:gen-mull-table-perm}
  Let $G$ and $H$ be groups, let $\cG_\ell$ and $\cH_\ell$ be the collections of all ordered generating sets of $G$ and $H$ of size at most $\ell$, and define $M_\ell(G) = \setb{M_{\bmg}}{\bmg \in \cG_\ell}$.  Then

  \begin{enumerate}
  \item If $G \not\cong H$, then $M_\ell(G) \cap M_\ell(H) = \emptyset$.
  \item If $G \cong H$, then $M_\ell(G) = M_\ell(H)$.
  \end{enumerate}
\end{lemma}

\begin{proof}
  For part (a), suppose $G \not\cong H$ but $M \in M_\ell(G) \cap M_\ell(H)$.  Then $G$ would be isomorphic to the group defined by the multiplication table $M$ which is also isomorphic to $H$.

  For part (b), fix an isomorphism $\phi : G \ra H$.  We claim that $M_{\bmg} = M_{\phi(\bmg)}$ for each $\bmg \in \cG_\ell$.  We know from \lemref{gen-ord} that for $x, y \in G$, $x \prec_{\bmg} y$ \ifft $\phi(x) \prec_{\phi(\bmg)} \phi(y)$.  Since $\phi(x) \phi(y) = \phi(x y)$, it follows that $M_{\bmg} = M_{\phi(\bmg)}$.  Therefore, $M_\ell(G) = M_\ell(H)$.
\end{proof}

\inpub{Recall that the rank of a group is the size of a minimal generating set.}
 
\begin{corollary}
  \label{cor:gen-enum-can}
  Let $G$ be a group.  Then we can compute a canonical form for $G$ in $n^{\rank(G) + O(1)}$ time.
\end{corollary}

\begin{proof}
  We first determine the rank of $G$ in $n^{\rank(G) + O(1)}$ time by brute force.  Then we compute the set $\cG_{\rank(G)}$ and choose $\can_{\grp}(G)$ to be the element that comes first lexicographically.  The fact that the map defined by this computation is a canonical form is immediate from \lemref{gen-mull-table-perm}.
\end{proof}

It is now easy to adapt \thmref{p-group-iso} to perform $p$-group canonization.

\begin{theorem}
  \label{thm:p-group-can}
  $p$-group canonization is in $n^{\min\{(1 / 2) \log_p n + O(p),\;\log_p n\}}$ time.
\end{theorem}

\begin{proof}
  Let $G$ be a $p$-group.  Combining Theorems~\ref{thm:group-red-comp-can} and~\ref{thm:alpha-comp-can} yields an $n^{(1 / 2) \log_p n + O(p)}$ time algorithm for group canonization while \corref{gen-enum-can} gives an $n^{\log_p n + O(1)}$ time algorithm.  The result then follows from the same argument used in the proof of \thmref{p-group-iso}.
\end{proof}

\section*{Acknowledgements}
We thank Laci Babai for suggesting the simplified composition series reduction and other comments, Paul Beame and Aram Harrow for useful discussions and feedback, Joshua Grochow for additional references, and Paolo Codenotti, Johannes K\"{o}bler, Youming Qiao, Jacobo Tor\'{a}n and the anonymous reviewers for helpful comments.  Part of this work was completed while David Rosenbaum was a long term visitor at the Center for Theoretical Physics at the Massachusetts Institute of Technology.  David Rosenbaum was funded by the DoD AFOSR through an NDSEG fellowship; partial support was provided by the NSF under grant CCF-0916400.

\inpriv{\inlong{
    \appendix
    \newpage
    \section{$X$ and $Y$ are a category equivalence}
    \label{app:X-Y-equiv}
    \begin{theorem}
  The functors $X : \comp \ra \comptree$ and $Y : \comptree \ra \comp$ form a category equivalence.
\end{theorem}

\begin{proof}
  We start by noting that $YX = I_{\comp}$ so $YX$ is naturally isomorphic to the identity functor.  Let $A \in \comptree$ be a graph and let $S$ be a composition series such that $A \cong X(S)$.  Then $XY(A) \cong XYX(S)$ so $XY(A) \cong X(S)$ and $A \cong XY(A)$.  Then for each graph $A \in \comptree$, choose an isomorphism $\pi_A : XY(A) \ra A$.  We note that $XYXY = XY$ so that $XY(\pi) : XY(A) \ra XY(A)$ is an automorphism.  Define $\eta_A = \pi_A \circ XY(\pi^{-1})$; we claim that $\eta : XY \ra I_{\comptree}$ is a natural isomorphism.  Clearly, each $\eta_A$ is an isomorphism.  We need to show that for each $A, A' \in \comptree$ and every isomorphism $\theta : A \ra A'$, $\eta_{A'} \circ XY(\theta) = \theta \circ \eta_A$.  Since $XY$ is fully faithful by Theorems~\ref{thm:X-fff} and~\ref{thm:Y-fff}, this is equivalent to $XY(\eta_{A'}) \circ XYXY(\theta) = XY(\theta) \circ XY(\eta_A)$ which holds since $XY(\eta_A) = \id_{XY(A)}$.
\end{proof}
}}

\inlong{\newpage}
\bibliographystyle{initials}
\bibliography{$HOME/LaTeX/computer-science-references,$HOME/LaTeX/math-references,$HOME/LaTeX/quantum-computing-references} 

\end{document}